\documentclass[conference]{IEEEtran}
\IEEEoverridecommandlockouts
\usepackage{graphicx}
\usepackage{array}
\usepackage{amsmath}
\usepackage{float}
\usepackage{multirow}
\usepackage{array}
\usepackage{cancel}

\usepackage{tabu}
\usepackage{amsthm}
\usepackage{epsfig}
\usepackage{epstopdf}
\usepackage{epsf}
\usepackage{color}
\usepackage[english]{babel}

\theoremstyle{definition}
\newtheorem{theorem}{Theorem}
\newtheorem{lemma}[theorem]{Lemma}

\newtheorem{proposition}{Proposition}

\usepackage{amssymb}
\usepackage{cite}
\usepackage{amsmath,amssymb,amsfonts}
\usepackage{graphicx}
\usepackage{textcomp}
\usepackage{xcolor}
\usepackage{algorithm}
\usepackage[noend]{algpseudocode}
\usepackage[T1]{fontenc}
\usepackage{graphicx,lipsum,afterpage,subcaption}
\usepackage{enumitem}
\usepackage{lipsum}
\usepackage{mathtools}
\usepackage{cuted}

\begin{document}

	\title
	{Communication-Channel Optimized Partition}

	\author{\IEEEauthorblockN{Thuan Nguyen}
		\IEEEauthorblockA{School of Electrical and\\Computer Engineering\\
			Oregon State University\\
			Corvallis, OR, 97331\\
			Email: nguyeth9@oregonstate.edu}
		\and
		\IEEEauthorblockN{Thinh Nguyen}
		\IEEEauthorblockA{School of Electrical and\\Computer Engineering\\
			Oregon State University\\
			Corvallis, 97331 \\
			Email: thinhq@eecs.oregonstate.edu}
	}

	\maketitle
	\begin{abstract}
Given an original discrete source $X$ with the distribution $p_X$ that is corrupted by noise to produce the noisy data $Y$ with the given joint distribution $p_{(X,Y)}$. A quantizer/classifier $Q: Y \rightarrow Z$ is then used to classify/quantize the data $Y$ to the discrete partitioned output $Z$ with probability distribution $p_Z$. Next, $Z$ is transmitted over a deterministic channel with a given channel matrix $A$ that produces the final discrete output $T$. One wants to design the optimal quantizer/classifier $Q^*$ such that the cost function $F(X,T)$ between the input $X$ and the final output $T$ is minimized while the probability of the partitioned output $Z$ satisfies a concave constraint $G(p_Z) \leq C$. Our results generalized some famous previous results. First, an iteration linear time complexity algorithm is proposed to find the local optimal quantizer. Second, we show that the optimal partition should produce a hard partition that is equivalent to the cuts by hyper-planes in the space of the posterior probability $p_{X|Y}$. This result finally provides a polynomial-time algorithm to find the globally optimal quantizer. 
	\end{abstract}
	Keyword: partition, channel quantization, impurity, optimization, constraints, mutual information, entropy.
	\section{Introduction}
Channel optimized partition/quantization is a common approach to lossy-compression data source-channel coding  that aims to minimize the end to end distortion when the quantized/classified data is transmitted over a noisy channel. Due to the huge volume of data and the limited rate of the transmission channel,  the data should be coded/quantized at the local stations/nodes before transmitted over a channel to the central station/node. The quality of the relay channel that is specified by its channel matrix, therefore, is important. Of course, one should design the partition/classification based on the channel matrix of the relay channel. 
From the source coding perspective, the quality of quantization/partition is normally measured by the end-to-end distortion between the input and the final output. While the squared-error distortion often uses to measure the distortion of scalar quantization, it is less appropriate for other problems in communication context i.e., maximizing the mutual information or minimizing the compression rate where other distortions i.e., the Kullback-Leiber divergence is more preferred. 
From the channel coding perspective, one should design the quantizer such that the compression rate of partition output is smaller than the channel capacity. From the power consumption perspective, the partitioned output should be coded such that the total energy consumption is below the power budget of transmitters. Generally, one has to design the optimal quantizer such that the partitioned output has to satisfy a certain constraint while an end-to-end cost function between the input and the final output is minimized. 

In this paper, we consider the design of quantizer with the aim of minimizing the end-to-end impurity between the input and the final output while the probability distribution of the partitioned output satisfies a certain concave constraint. The impurity termed the loss function that measures the "impurity" of the partitioned sets. Some of the popular impurity functions are the entropy function and the Gini index \cite{breiman2017classification}, \cite{quinlan2014c4}. For example, when the empirical entropy of a set is large, this indicates a high level of non-homogeneity of the elements in the set, i.e., "impurity". Impurity function was vastly used in learning theory and decision tree \cite{quinlan2014c4}, \cite{breiman2017classification}, \cite{nadas1991iterative}, \cite{chou1991optimal}, \cite{burshtein1992minimum}, \cite{coppersmith1999partitioning}.
Interestingly, if the impurity is conditional entropy, minimizing impurity is equivalent to maximizing the mutual information between the input and the final output \cite{kurkoski2014quantization},
\cite{zhang2016low}. Therefore, partition/quantization that minimizes the entropy impurity has many applications in communication \cite{kurkoski2014quantization}, \cite{kurkoski2017single}, \cite{romero2015decoding}, \cite{vangala2015quantization}. 
On the other hand, design the optimal partition such that the partitioned output has to satisfy a constraint is very important in the case of the relay channel is a limited resource channel. For example, if the relay channel is a low bandwidth channel, the entropy of partitioned output that controls the maximum compression transmission rate is very important. The power and time delay of transmission constraints also can be constructed similarly to the entropy constraint to establish some useful applications. That said, the problem of finding the optimal quantizer that minimizes the end-to-end impurity between input and final output under a constraint is an interesting problem  that covers many sub-problems in \cite{kurkoski2014quantization}, \cite{zhang2016low}, \cite{nguyen2018capacities}, \cite{strouse2017deterministic}, \cite{winkelbauer2013channel}. For example, if it is non-constraint with partitioned output and the channel matrix is an identity matrix, our setting is back to the model in \cite{kurkoski2014quantization}, \cite{zhang2016low}, \cite{nguyen2018capacities} using the impurity function is conditional entropy. If the channel matrix is not an identity matrix  and the impurity function is conditional entropy, our problem can be viewed as the problem in \cite{winkelbauer2013channel}. If the channel matrix is an identity matrix and there isn't any constraint for partitioned output, our setting is identical to the setting in \cite{laber2018binary} using Gini index impurity function. 
Finally, if the relay channel is perfect (channel matrix is an identity matrix)  and both impurity and constraint function are entropy, our problem is the same as the problem in \cite{strouse2017deterministic}. The more detail of these sub-problems can be seen in Section II. 

The outline of our paper is as follows.  In Section \ref{sec: formulation}, we describe the problem formulation and its applications.  In Section \ref{sec: solution}, we provide the optimality condition for the optimal partition.  In Section \ref{sec: algorithm}, we provide an iteration algorithm that can find the local optimal solution. Moreover, we show that the optimal partition is equivalent to the cuts by hyper-planes in the probability space of the posterior probability. Finally,  we provide a few concluding remarks in Section \ref{sec: conclusion}.

\section{Problem Formulation}
	\label{sec: formulation}
		\begin{figure}
		\centering
		\includegraphics[width=2.5 in]{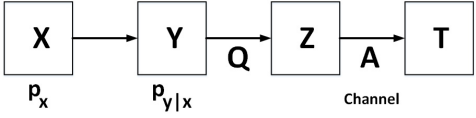}\\
		\caption{The quantizer $Q$ is designed to minimize the impurity function between input $X$ and final output $T$ while the partitioned output $Z$ has to satisfy a certain constraint.}\label{fig: 1}
	\end{figure}	
	

	Fig. \ref{fig: 1} illustrates our model.   The input set consists of $N$ discrete symbols $X=(X_1,X_2,\dots,X_N)$  with a given pmf $p_X=\{p_1,p_2,\dots,p_N \}$. The data set consists of $M$ discrete vectors $Y=(Y_1,Y_2,\dots,Y_M)$ having the pmf $p_Y=\{p_{Y_1},p_{Y_2},\dots,p_{Y_M}\}$ and the joint distribution $p_{(X_n,Y_m)}$, $\forall$ $n=1,2,\dots,N$ and $m=1,2,\dots,M$. 
$Y$ will be quantized to produce the partitioned output $Z=(Z_1,Z_2,\dots,Z_K)$ having the pmf $p_Z=\{p_{Z_1},p_{Z_2},\dots,p_{Z_K}\}$ using a quantizer $Q: Y \rightarrow Z$. Noting that $Q$ is possible a stochastic quantizer i.e., $0 \leq p_{Z_k|Y_m} \leq 1$.  The partitioned output $Z$ is then transmitted over a relay channel having channel matrix $A$  to produce the final output $T=(T_1,T_2,\dots,T_H)$.
Noting that the entry $A_{kh}$ of channel matrix $A$ denotes the conditional probability  $p_{T_h|Z_k}$ such that the transmitter transmits $Z_k$ but the receiver received $T_h$, i.e., $A_{kh}=p_{T_h|Z_k}$ for $h=1,2,\dots,H$ and $k=1,2,\dots,K$. 	Our goal is finding the optimal quantizer $Q^*$ to minimize the end-to-end impurity/cost function between input and the final output $F(X,T)$ while the partitioned output satisfies a certain constraint $G(p_Z) \leq C$. 

\subsection{Cost measurement}	
We consider the impurity/cost function that takes the following form $F(X,T) = \sum_{h=1}^{H} F(X,T_h),$
where $F(X,T_h)$ denotes the impurity in each final output $T_h$.  
\begin{equation}
F(X,T_h) =\sum_{h=1}^{H} p_{T_h}f[p_{X_1|T_h}, p_{X_2|T_h},\dots, p_{X_N|T_h}]. \label{eq: cost function 1}
\end{equation}

That said, the total impurity $F(X,T)$ is added up from the impurity in each final output $F(X,T_h)$. The factor $p_{T_h}$ denotes the weight of the final output $T_h$, $f[.]$ is a \textit{concave} function that measures the loss in each final output $T_h$ and $p_{X_n|T_h}$ denotes the conditional probability of $X_n$ given $T_h$. For convenient, let's define 
\begin{equation*}
p_{(X,Y_m)}=[p_{(X_1,Y_m)}, p_{(X_2,Y_m)}, \dots, p_{(X_N,Y_m)}],
\end{equation*}

\begin{equation}
\label{eq: 5}
p_{(X,T_h)}=[p_{(X_1,T_h)}, p_{(X_2,T_h)}, \dots, p_{(X_N,T_h)}].
\end{equation}

\begin{equation*}
p_{X|T_h}=[p_{X_1|T_h}, p_{X_2|T_h}, \dots, p_{X_N|T_h}],
\end{equation*}

Now, suppose that a quantizer $Q$ quantizes  $Q(Y_m) \rightarrow Z_k$ with the probability $p_{Z_k|Y_m}$, then
\begin{equation}
p_{(X_n,Z_k)}=\sum_{Y_m \in Y }^{} p_{(X_n,Y_m)} p_{Z_k|Y_m}. 
\end{equation}

However, the final output $T$ can be computed via the partitioned output $Z$ and the given channel matrix $A$.  Thus, $p_{(X_n,T_h)}$ can be determined by:
\begin{equation}
\label{eq: XZ and XT}
p_{(X_n,T_h)}=\sum_{k=1}^{K}p_{(X_n,Z_k)}A_{kh}.
\end{equation}

Now,  the impurity function in each final output $T_h$ can be rewritten by:
\begin{small}
\begin{eqnarray}
 F(X,T_h) \!=\!   (\sum_{n\!=\!1}^{N}p_{(X_n\!,\!T_h)})f[\dfrac{p_{(X_1\!,\!T_h)}}{\sum_{n\!=\!1}^{N}p_{(X_n\!,\!T_h)}}, \dots,\dfrac{p_{(X_N\!,\!T_h)}}{\sum_{n\!=\!1}^{N}p_{(X_n\!,\!T_h)}}  ]
\end{eqnarray}
\end{small}
where $\sum_{n=1}^{N}p_{(X_n,T_h)}$ is the weight of $T_h$ and $\dfrac{p_{(X_n,T_h)}}{\sum_{n=1}^{N}p_{(X_n,T_h)}}$ denotes the conditional distribution $p_{X_n|T_h}$. \textit{The impurity function, therefore, is only the function of the joint distribution $p_{(X_n,T_h)}$.} For convenient, in the rest of paper, we denote $F(X,T_h)$ by $F(p_{(X,T_h)})$ where the joint distribution vector $p_{(X,T_h)}$ is defined in (\ref{eq: 5}). 

\subsection{Constraints of the partitioned output}
We want to design the quantizer such that the partitioned output satisfies the following constraint
\begin{equation*}
G(p_Z)=\sum_{k=1}^{K}g_k(p_{Z_k}) \leq C
\end{equation*}
where $g_k(.)$ is an arbitrary \textit{concave} function, $\forall$ $k$, i.e., the entropy function, the linear function. For example, if we want to compress data $Y$ to $Z$ and then transmit $Z$ over a low bandwidth channel, the entropy of $p_Z$ which is controlled the maximum compression rate, is important. Similarly, to transmit $Z$ over a channel, each value $Z_k$ is coded to a pulse which have a difference cost of transmission i.e., power consumption or time delay. The cost of transmission now can be formulated by a linear constraint.  
\subsection{Problem Formulation}
To jointly design the quantizer such that the impurity function is minimized while the partitioned output satisfies a certainty constraint,   we are interested in solving the following optimization problem:
	
	\begin{equation}
		\label{eq: jointly optimization} 
		Q^*=\min_{Q} \beta F(X;T) + G(p_Z),
	\end{equation}
where $\beta >0$ is pre-specified parameter to control the trade-off between minimizing $F(X;T)$ or minimizing $G(p_Z)$. Noting that corresponding to the setting of $f[.]$, $g(.)$, $\beta$ and channel matrix $A$, our problem generalized many sub-problems. For example, if $f[.]$ is entropy function, $\beta=+\infty$ and $A$ is an identity matrix, we solve the problem in \cite{kurkoski2014quantization}, \cite{zhang2016low}. If $f[.]$  is Gini index or entropy, $A$ is an identity matrix, $N=2$ and $\beta=+\infty$, the problem is solved in \cite{laber2018binary}, if $f[.]$ is entropy function, $\beta=+\infty$ and $A$ is an identity matrix, the problem in \cite{winkelbauer2013channel} is solved. If both impurity and constraint are entropy and $A$ is an identity matrix, our setting is identical to the setting in \cite{strouse2017deterministic}. 

Noting that we assume that both $f[.]$ and $g_k[.]$ are \textit{concave} functions which satisfy the following inequality:
\begin{equation}
\label{eq: concave function}
f(\lambda a + (1-\lambda)b) \geq \lambda f(a) + (1-\lambda)f(b), \forall \lambda \in (0,1),
\end{equation}
for all probability vector $a=[a_1,a_2,\dots,a_N]$ and $b=[b_1,b_2,\dots,b_N]$  with equality if and only if $a=b$.
Based on the concave property, an iteration algorithm is proposed to find the local optimal quantizer. Moreover, we show that the optimal quantizers (local and global) produce a hard partition that is equivalent to the cuts by hyper-planes in the space of the posterior probability $p_{X|Y}$. This interesting property finally yields a polynomial time algorithm to determine the truly global optimal quantizer. 

\section{Optimality Condition}
\label{sec: solution}

We first begin with some properties of $F(X,T_h)$. 
\begin{proposition}
\label{prop: 2}
The impurity in each subset $T_h$ is defined by $F(X,T_h)$ which has the following properties:

(i) The impurity function is proportional increasing/ decreasing to its weight: if  $p_{(X,T_a)}=\lambda p_{(X,T_{b})}$, then
\begin{equation}
\dfrac{F(X,T_a)}{F(X,T_b)}=\lambda.
\end{equation}

(ii) The impurity gain after partition is always non-negative: If $p_{(X,T_a)}=p_{(X,T_b)}+p_{(X,T_c)}$, then
\begin{equation}
\label{eq: concave of partition}
F(X,T_a) \geq F(X,T_b) + F(X,T_c).
\end{equation}

\end{proposition}
\begin{proof}

(i) From $p_{(X,T_a)}=\lambda p_{(X,T_b)}$, then  $p_{X|T_a}=p_{X|T_b}$ and $p_{T_a}=\lambda p_{T_b}$. Thus, using the definition of $F(X,T_h)$ in (\ref{eq: cost function 1}), it is obviously to prove the first property. 

(ii) By dividing both side of  $p_{(X,T_a)}=p_{(X,T_b)}+p_{(X,T_c)}$ to $p_{T_a}$, we have
\begin{equation}
\label{eq: 9}
p_{X|T_a}= \dfrac{p_{T_b}}{p_{T_a}}p_{X|T_b}+ \dfrac{p_{T_c}}{p_{T_a}}p_{X|T_c}.
\end{equation}

Now, using the original definition in (\ref{eq: cost function 1}),
\begin{eqnarray}
F(X,T_a) &\!=\!& p_{T_a}f(p_{X|T_a}) \nonumber\\
&\!=\!& p_{T_a} f [\dfrac{p_{T_b}}{p_{T_a}}p_{X|T_b} + \dfrac{p_{T_c}}{p_{T_a}}p_{X|T_c}] \label{eq: 10}\\
&\!\geq \!& p_{T_a} [\dfrac{p_{T_b}}{p_{T_a}} f(p_{X|T_b}) + \dfrac{p_{T_c}}{p_{T_a}}f(p_{X|T_c})] \label{eq: 11}\\
&\!=\!&p_{T_b}f(p_{X|T_b}) + p_{T_c}f(p_{X|T_c}) \nonumber\\
&\!=\!& F(X,T_b) + F(X,T_c) \nonumber
\end{eqnarray}
with (\ref{eq: 10}) is due to (\ref{eq: 9}) and (\ref{eq: 11}) due to concave property of $f(.)$ which is defined in (\ref{eq: concave function}) using $\lambda= \dfrac{p_{T_b}}{p_{T_a}}$, $1-\lambda=\dfrac{p_{T_c}}{p_{T_a}}$. 
\end{proof}

Now, we are ready to show the main result which characterizes the condition for an  optimal partition $Q^*$. 
\begin{theorem}
\label{theorem: 1}
Suppose that  an optimal quantizer $Q^*$ yields the optimal partitioned output  $Z=(Z_1,Z_2,\dots,Z_K )$ and the optimal final output $T=(T_1,T_2,\dots,T_H )$. We define vector $c_k=[c_k^1,c_k^2,\dots,c_k^N]$, $k=1,2,\dots,T$ where
\begin{equation}
\label{eq: 16}
c_k^n= \frac{\partial F(p_{(X,T_k)})}{\partial p_{(X_n,T_k)}}, \forall n \in \{1,2,\dots,N\}. 
\end{equation}
We also define 
\begin{equation}
\label{eq: 17}
d_k=\frac{\partial g_k(p_{Z_k})}{\partial p_{Z_k}}.
\end{equation}
Define the "distance" from $Y_m \in Y$ to $Z_k$ is 
\begin{eqnarray}
D(Y_m,Z_k) &\!=\!&\beta \sum_{h=1}^{H} \sum_{n=1}^{N}[c_k^n p_{(X_n,Y_m)} ]A_{kh} \!+\! d_k p_{Y_m}. \label{eq: optimality condition}
\end{eqnarray}
Then, data $Y_m$ is quantized to $Z_k$ if and only if $D(Y_m,Z_k) \leq D(Y_m,Z_{s})$
for $\forall s \in \{1,2,\dots,K\}$ and $ s \neq k$. 
\end{theorem}
\begin{proof}
Due to the limited space, we only provide the outline of proof. Consider two arbitrary optimal partitioned outputs $Z_q$ and $Z_s$ and a trial data $Y_m$. For a given optimal quantizer $Q^*$, we suppose that $Y_m$ is allocated to $Z_q$ with the probability of $p_{Z_q|Y_m}=v$, $0 < v \leq 1$ (soft partition). We remind that $p_{(X,Y_m)}=[p_{(X_1,Y_m)},p_{(X_2,Y_m)}, \dots, p_{(X_N,Y_m)}]$ denotes the joint distribution in the sample $Y_m$. We will determine the change of the impurity function $F(X,T)$ and the constraint  $G(p_Z)$ as a function of $t$  when changing amount of $tvp_{(X,Y_m)}$ from $p_{(X,Z_q)}$ to $p_{(X,Z_s)}$ where $t$ is a scalar and $0<t<1$. By changing $tvp_{(X,Y_m)}$, the new joint distributions in $Z_q$ and $Z_s$ are $p_{(X,Z_q)}-tvp_{(X,Y_m)}$ and $p_{(X,Z_s)}+tvp_{(X,Y_m)}$, respectively. 
Thus, from (\ref{eq: XZ and XT}), the new joint distribution in $T_h$ as a function of $t$ is ${p_{(X,T_h)}}_t$  can be determined by:
\begin{eqnarray*}
{p_{(X,T_h)}}_t &=& p_{(X,T_h)} -tvp_{(X,Y_m)} A_{qh} + tvp_{(X,Y_m)}A_{sh}\\
&=& p_{(X,T_h)} +  tvp_{(X,Y_m)} (A_{sh}-A_{qh}).
\end{eqnarray*}
Now, denote $tvp_{(X,Y_m)} (A_{sh}-A_{qh}) =\delta_{th}$. The total change of impurity function $F(X,T)$ and constraint $G(p_Z)$ are:
\begin{eqnarray}
F(X,T)_t &\!=\!&  \sum_{h=1}^{H} F(p_{(X,T_h)} + \delta_{th} ) \nonumber\\
G(p_Z)_t &=& \sum_{k=1,k \neq q,s}^{K}g_k(p_{Z_k}) \nonumber\\
&+& g_q(p_{Z_q}-tvp_{Y_m}) + g_s(p_{Z_s} + tvp_{Y_m}) \nonumber.  
\end{eqnarray}
The total instantaneous change of $\beta F(X,T) + G(p_Z)$ as a function of $t$ is
\begin{small}
\begin{eqnarray}
I_t \!&\!=\!& \beta \! [\sum_{h\!=1}^{H} \! F(p_{(X,T_h)} \!+\! \delta_{th} )] \!+\!  g_q(p_{Z_q}\!-\!tvp_{Y_m}) \!+\! g_s(p_{Z_s} \!+\! tvp_{Y_m}). \nonumber \\
\label{eq: I(t)}
\end{eqnarray}
\end{small}
However,
\begin{equation}
\label{eq: derivative 1}
 \frac{\partial F(X,T)_t}{\partial t}|_{t=0}=  v \beta \sum_{h=1}^{H} \sum_{n=1}^{N} (c_k^n p_{(X_n,Y_m)}) (A_{sh}-A_{qh}),
\end{equation}
\begin{equation}
\label{eq: derivative 3}
   \frac{\partial G(p_Z)_t}{\partial t}|_{t=0}=v [d_s p_{Y_m}-d_q p_{Y_m}].
\end{equation}
From (\ref{eq: I(t)}), (\ref{eq: derivative 1}), (\ref{eq: derivative 3}) and (\ref{eq: optimality condition}), we have
\begin{eqnarray}
\frac{\partial I_t}{\partial t}|_{t=0}=v[D(Y_m,Z_s)-D(Y_m,Z_q)]. \nonumber
\end{eqnarray}
Now, using contradiction method, suppose that $D(Y_m,Z_q) > D(Y_m,Z_s)$. Thus,
\begin{equation}
\label{eq: 20}
\frac{\partial I_t}{\partial t} |_{t=0} < 0.
\end{equation}

\begin{proposition}
\label{prop: 1}
Consider $I_t$ which is defined in (\ref{eq: I(t)}). For $0 < t < a < 1$, we have:
\begin{equation}
\label{eq: gradient of I(t)}
\dfrac{I_t-I_0}{t} \geq \dfrac{I_a-I_0}{a}.
\end{equation}
\end{proposition}
\begin{proof}
 Due to the limited space, we sketch the proof as following. First, 
(\ref{eq: gradient of I(t)}) is equivalent to:
\begin{equation}
\label{eq: vyha1}
I_t \geq (1-\dfrac{t}{a})I_0 + \dfrac{t}{a}I_a.
\end{equation}
Noting that $I_t$ is the combination of the impurity function $F(p_{(X,T_h)} + \delta_{th})$ and the constraint functions $g_q(.)$, $g_s(.)$ that admit the concavity properties in Proposition \ref{prop: 2} and equation (\ref{eq: concave function}).  By using a little bit of algebra, one can show that (\ref{eq: vyha1}) follows by the concavity properties that finally proves (\ref{eq: gradient of I(t)}). Please see the full proof in our extension version.  
\end{proof}

Now, we continue to the proof of Theorem \ref{theorem: 1}. From Proposition \ref{prop: 1} and the assumption in (\ref{eq: 20}), we have:
\begin{equation*}
 0 > \frac{\partial I_t}{\partial t} |_{t=0}=\lim \dfrac{I_t-I_0}{t} \geq \dfrac{I_{1}-I_0}{1}.
\end{equation*}

Thus, $I_0>I_1$ which obviously implies that by completely changing amount of $v p_{(X,Y_m)}$  from $p_{(X,Z_q)}$ to $p_{(X,Z_s)}$, the total of the loss is obviously reduced. This contradicts to our assumption that $Q^*$ is an optimal quantizer. By contradiction method, the proof is complete. 
\end{proof}

\begin{lemma}
\label{lemma: 1}
The optimal quantizer of the problem (\ref{eq: jointly optimization}) is a
deterministic quantizer (hard clustering) i.e., $p_{Z_i|Y_j} \in \{0,1\}$, $\forall$ $i,j$.  
\end{lemma}

\begin{proof}
Due to the limited space, we do not give the full proof. However, based on the proof of Theorem \ref{theorem: 1}, one can easily verify that if quantizer $Q$ only allocates a part of $p_{(X,Y_m)}$ to $p_{(X,Z_q)}$, i.e., distribute $v p_{(X,Y_m)}$ to $p_{(X,Z_q)}$ for $0 <v <1$ (soft partition), then $Q$ is not optimal. The reason is that if the distance from $D(Y_m,Z_s)$ is shortest, the impurity can be reduced by completely moving  $v p_{(X,Y_m)}$ from $Z_q$ to $Z_s$ i.e., $p_{Z_s|Y_m}=1$. That said,  the optimal partition is hard partition. 
\end{proof}

\section{Algorithms}
 \label{sec: algorithm}
\subsection{Practical Algorithm}
From the optimality condition in Theorem \ref{theorem: 1},  we should allocate the data $Y_m$ to the partitioned output $Z_k$ if and only if the "distance" $D(Y_m,Z_k)$ is shortest. Therefore, a simple alternative optimization algorithm that is very similar to the k-means algorithm can be applied to find the locally optimal solution. Our algorithm is proposed in Algorithm \ref{alg: 1}. We also note that the distance $D(Y_m, Z_k)$ is 
\begin{eqnarray*}
D(Y_m,Z_k) &\!=\!&\beta \sum_{h=1}^{H} \sum_{n=1}^{N}[c_k^n p_{(X_n,Y_m)} ]A_{kh} \!+\! d_k p_{Y_m} \nonumber\\
			   &=& p_{Y_m} [\beta \sum_{h=1}^{H} \sum_{n=1}^{N}[c_k^n p_{X_n|Y_m} ]A_{kh} \!+\! d_k].
\end{eqnarray*}
Therefore, one can ignore the constant $p_{Y_m}$ while comparing the distance $D(Y_m,Z_k)$ and use a simpler version distance $D'(Y_m,Z_k)$ as following
\begin{equation}
\label{eq: simpler distance}
D'(Y_m,Z_k)=\beta \sum_{h=1}^{H} \sum_{n=1}^{N}[c_k^n p_{X_n|Y_m} ]A_{kh} \!+\! d_k.
\end{equation}

\begin{footnotesize}
	\begin{algorithm}
		\caption{Communication Optimized Partition}
		\label{alg: 1}
		\begin{algorithmic}[1]
			\State{\textbf{Input}: $p_X$, $p_Y$, $p_{(X,Y)}$, $f(.)$, $g_k(.)$ and $\beta$.}
			\State{\textbf{Output}: $Z=\{Z_1,Z_2,\dots,Z_K$ \} }
			
			\State{\textbf{Initialization}: Randomly hard clustering $Y$ into $K$ clusters. }
			
			\State{\textbf{Step 1}: Updating $p_{(X,Z_k)}$, $p_{(X,T_h)}$ and $d_k$ for $\forall$ $k \in \{1,2,\dots,K\}$ and $h \in \{1,2,\dots,H\}$:}
		\begin{equation*}
		p_{(X_n,Z_k)}= \sum_{Y_m \in Z_k}^{}p_{(X_n,Y_m)},
		\end{equation*}
		\begin{equation*}
p_{(X_n,T_h)}=\sum_{k=1}^{K}p_{(X_n,Z_k)}A_{kh},
		\end{equation*}
		\begin{equation*}
		c_k^n= \frac{\partial F(p_{(X,T_k)})}{\partial p_{(X_n,T_k)}}, \forall n \in \{1,2,\dots,N\}, 
		\end{equation*}
		\begin{equation*}
		p_{Z_k} =\sum_{Y_m \in Z_k}p_{Y_m},
		\end{equation*}			
		\begin{equation*}
		d_k= \frac{\partial g_k(p_{Z_k})}{\partial p_{Z_k}}. 
		\end{equation*}
			\State{\textbf{Step 2}: Updating the membership by measurement the distance from each $Y_m \in Y$ to each $Z_k \in Z$}
\begin{equation}
\label{eq: nearest clustering}
Z_k=\{Y_m| D(Y_m,Z_k) \leq D(Y_m,Z_s), \forall s \neq k,
\end{equation}
where $D(Y_m,Z_k)$ is defined in (\ref{eq: optimality condition}) or in (\ref{eq: simpler distance}). 
			\State{\textbf{Step 3}: Go to Step 1 until all partitioned outputs $\{Z_1,Z_2,\dots,Z_K\}$ stop changing or the maximum number of iterations has been reached.}				
		\end{algorithmic}
	\end{algorithm}
\end{footnotesize}

The Algorithm \ref{alg: 1} works similarly to the k-means algorithm and the distance from each point in $Y$ to each partitioned output in $Z$ is updated over each loop. The complexity of this algorithm, therefore, is $O(TNKM)$ where $T$ is the number of iterations, $N$, $K$, $M$ are the size of data dimensional, the size of partitioned set $Z$ and the size of data set $Y$. 

\subsection{Hyper-plane separation}
\label{subsection: 3-C}

Similar to the work in \cite{burshtein1992minimum}, in this paper, we show that the optimal partition is equivalent to the cuts by hyper-planes in the space of the posterior probability. Therefore, existing a polynomial time algorithm that can find the globally optimal quantizer. Indeed, consider the optimal quantizer $Q^*$ that produces a given optimal partition output $Z=\{Z_1,Z_2,\dots,Z_K\}$. From the optimality condition in Theorem \ref{theorem: 1}, we know that $\forall$ $Y_m \in Z_k$, then $D(Y_m,Z_k) \leq D(Y_m,Z_s)$ for $\forall$ $s \neq k$. Now, using the distance in (\ref{eq: simpler distance}), we have
\begin{equation*}
\beta \sum_{h=1}^{H} \sum_{n=1}^{N}[c_k^n p_{X_n|Y_m} ]A_{kh} \!+\! d_k \leq \beta \sum_{h=1}^{H} \sum_{n=1}^{N}[c_s^n p_{X_n|Y_m} ]A_{sh} \!+\! d_s,
\end{equation*}
or
\begin{eqnarray*}
0 &\geq &  \beta \sum_{h=1}^{H} \sum_{n=1}^{N} [c_k^nA_{kh}-c_s^n A_{sh}] p_{X_n|Y_m} + d_k-d_s.
\end{eqnarray*}
From $p_{(X_N|Y_m)}=1-\sum_{n=1}^{N-1} p_{(X_n|Y_m)}$, then
\begin{small}
\begin{eqnarray}
\label{eq: hyperplane}
0 \!&\! \geq \!&\!  \beta \! \sum_{h=1}^{H} \sum_{n=1}^{N-1} [(c_k^nA_{kh } \! - \! c_s^n A_{sh}) - (c_k^N A_{kh} \!-\! c_s^N A_{sh})] p_{(X_n|Y_m)} \nonumber\\
\! &\!-\!&\! [d_s\!-\!d_k\!+\!\beta \! \sum_{h=1}^{H} (c_s^N A_{sh}\!-\!c_k^N A_{kh})].
\end{eqnarray}
\end{small}
For a given optimal quantizer $Q^*$, $c_k^n$ ,$c_s^n$, $d_k$, $d_s$ are all scalars $\forall$ $n,k,s$. Thus, equation (\ref{eq: hyperplane}) is equivalent to a hyper-plane in the $N-1$ dimensional probability space that can be constructed by using posterior probability $p_{X_n|Y_m}$ $\forall$ $n=1,2,\dots,N-1$. That said, all of $Y_m \in Z_k$ is separated by a hyper-plane cut in $N-1$ dimensional probability space of posterior probability $p_{X_n|Y_m}$. Similar to the results proposed in \cite{burshtein1992minimum}, existing a polynomial time  algorithm having time complexity of $O(M^{N})$ which can exhausted searching all the hyper-plane cuts that finally provides the globally optimal quantizer. 

\subsection{Discussion and Application}

Due to the limited space, we will not provide numerical results in this paper. Instead, using the property of hyper-plane separation, we show that a polynomial time algorithm having the complexity of $O(M^3)$ is able to find the globally optimal quantizer if the input source is binary. Similar to the work in \cite{kurkoski2014quantization}, if $N=2$, then a hyper-plane is a point in the probability space of posterior probability $p_{X|Y}$. Thus, the globally optimal quantizer can be found by considering only the convex cell quantizer in probability space, i.e., the optimal quantizer is a scalar quantizer in posterior probability variable $p_{X_1|Y}$.  The convex cell property can help to find the global optimal quantizer in a polynomial time complexity using dynamic programming. We refer the reader to the work in \cite{kurkoski2014quantization} for the detailed algorithm. The complexity of the traditional dynamic programming to find the globally optimal quantizer is $O(M^3)$ in the worst case. In \cite{iwata2014quantizer}, the time complexity of algorithm in \cite{kurkoski2014quantization} can be further reduced to a linear time complexity using SMAWK  algorithm \cite{aggarwal1987geometric}
As an open problem, we wonder that is it possible to using the same technique in \cite{iwata2014quantizer} to reduce the time complexity of our problem if the input source is binary?

\section{Conclusion}
 \label{sec: conclusion}
The problem of designing the optimal quantizer that minimizes the end-to-end impurity function between the input and the final output under a partitioned output constraint is investigated. Our results generalized some previous results.  An iteration algorithm was proposed to find the local optimal quantizer in a linear time complexity. In additional, we also show that the optimal quantizer produces a hard partition that is equivalent to hyper-plane cuts in the probability space of the posterior probability. Thus, there exists a polynomial time algorithm that can determine the globally optimal quantizer. Interestingly, if the input source is binary, a dynamic programming technique can be applied that is able to find the globally optimal solution in a cubic of time complexity.
\
\bibliographystyle{unsrt}
\bibliography{sample}

\begin{thebibliography}{10}

\bibitem{breiman2017classification}
Leo Breiman.
\newblock {\em Classification and regression trees}.
\newblock Routledge, 2017.

\bibitem{quinlan2014c4}
J~Ross Quinlan.
\newblock {\em C4. 5: programs for machine learning}.
\newblock Elsevier, 2014.

\bibitem{nadas1991iterative}
Arthur N{\'a}das, David Nahamoo, Michael~A Picheny, and Jeffrey Powell.
\newblock An iterative'flip-flop'approximation of the most informative split in
  the construction of decision trees.
\newblock In {\em [Proceedings] ICASSP 91: 1991 International Conference on
  Acoustics, Speech, and Signal Processing}, pages 565--568. IEEE, 1991.

\bibitem{chou1991optimal}
Philip~A. Chou.
\newblock Optimal partitioning for classification and regression trees.
\newblock {\em IEEE Transactions on Pattern Analysis \& Machine Intelligence},
  (4):340--354, 1991.

\bibitem{burshtein1992minimum}
David Burshtein, Vincent Della~Pietra, Dimitri Kanevsky, and Arthur Nadas.
\newblock Minimum impurity partitions.
\newblock {\em The Annals of Statistics}, pages 1637--1646, 1992.

\bibitem{coppersmith1999partitioning}
Don Coppersmith, Se~June Hong, and Jonathan~RM Hosking.
\newblock Partitioning nominal attributes in decision trees.
\newblock {\em Data Mining and Knowledge Discovery}, 3(2):197--217, 1999.

\bibitem{kurkoski2014quantization}
Brian~M Kurkoski and Hideki Yagi.
\newblock Quantization of binary-input discrete memoryless channels.
\newblock {\em IEEE Transactions on Information Theory}, 60(8):4544--4552,
  2014.

\bibitem{zhang2016low}
Jiuyang~Alan Zhang and Brian~M Kurkoski.
\newblock Low-complexity quantization of discrete memoryless channels.
\newblock In {\em 2016 International Symposium on Information Theory and Its
  Applications (ISITA)}, pages 448--452. IEEE, 2016.

\bibitem{kurkoski2017single}
Brian~M Kurkoski and Hideki Yagi.
\newblock Single-bit quantization of binary-input, continuous-output channels.
\newblock In {\em 2017 IEEE International Symposium on Information Theory
  (ISIT)}, pages 2088--2092. IEEE, 2017.

\bibitem{romero2015decoding}
Francisco Javier~Cuadros Romero and Brian~M Kurkoski.
\newblock Decoding ldpc codes with mutual information-maximizing lookup tables.
\newblock In {\em Information Theory (ISIT), 2015 IEEE International Symposium
  on}, pages 426--430. IEEE, 2015.

\bibitem{vangala2015quantization}
Harish Vangala, Emanuele Viterbo, and Yi~Hong.
\newblock Quantization of binary input dmc at optimal mutual information using
  constrained shortest path problem.
\newblock In {\em Telecommunications (ICT), 2015 22nd International Conference
  on}, pages 151--155. IEEE, 2015.

\bibitem{nguyen2018capacities}
Thuan Nguyen, Yu-Jung Chu, and Thinh Nguyen.
\newblock On the capacities of discrete memoryless thresholding channels.
\newblock In {\em 2018 IEEE 87th Vehicular Technology Conference (VTC Spring)},
  pages 1--5. IEEE, 2018.

\bibitem{strouse2017deterministic}
DJ~Strouse and David~J Schwab.
\newblock The deterministic information bottleneck.
\newblock {\em Neural computation}, 29(6):1611--1630, 2017.

\bibitem{winkelbauer2013channel}
Andreas Winkelbauer, Gerald Matz, and Andreas Burg.
\newblock Channel-optimized vector quantization with mutual information as
  fidelity criterion.
\newblock In {\em 2013 Asilomar Conference on Signals, Systems and Computers},
  pages 851--855. IEEE, 2013.

\bibitem{laber2018binary}
Eduardo~S Laber, Marco Molinaro, and Felipe A~Mello Pereira.
\newblock Binary partitions with approximate minimum impurity.
\newblock In {\em International Conference on Machine Learning}, pages
  2860--2868, 2018.

\bibitem{iwata2014quantizer}
Ken-ichi Iwata and Shin-ya Ozawa.
\newblock Quantizer design for outputs of binary-input discrete memoryless
  channels using smawk algorithm.
\newblock In {\em Information Theory (ISIT), 2014 IEEE International Symposium
  on}, pages 191--195. IEEE, 2014.

\bibitem{aggarwal1987geometric}
Alok Aggarwal, Maria~M Klawe, Shlomo Moran, Peter Shor, and Robert Wilber.
\newblock Geometric applications of a matrix-searching algorithm.
\newblock {\em Algorithmica}, 2(1-4):195--208, 1987.

\end{thebibliography}

\end{document}